\documentclass[11pt,letterpaper]{article}
\usepackage{setspace}

\usepackage[dvips]{graphicx}
\usepackage[cmex10]{amsmath}
\usepackage{amssymb}
\usepackage{amsthm}
\usepackage{amsfonts}
\usepackage{bm}
\usepackage{cite}
\usepackage[svgnames]{xcolor} 
\usepackage{pstricks,pst-node,pst-plot,pstricks-add}
\usepackage{algpseudocode}
\usepackage{algorithm,algcompatible}
\usepackage{subcaption}

\allowdisplaybreaks

\algnewcommand\INPUT{\item[\textbf{Input:}]}
\algnewcommand\OUTPUT{\item[\textbf{Output:}]}
\algnewcommand\INIT{\item[\textbf{Initialize:}]}
\algnewcommand\ON{\item[\textbf{on}]}
 \algnewcommand{\OR}{\textbf{or}}
  \algnewcommand{\AND}{\textbf{and}}

\newtheorem{theorem}{Theorem}

\newtheorem{remark}{Remark}

\newtheorem{corollary}[theorem]{Corollary}

\makeatletter
\newcommand\nocaption{%
    \renewcommand\p@subfigure{}
    \renewcommand\thesubfigure{\thefigure\alph{subfigure}}
}
\makeatother

\begin{document}
\title{Combating Computational Heterogeneity in Large-Scale Distributed Computing via Work Exchange}
\author{Mohamed Adel Attia \qquad Ravi Tandon \\
\normalsize Department of Electrical and Computer Engineering\\
\normalsize  University of Arizona, Tucson, AZ 85721\\
\normalsize E-mail:\em{\{madel, tandonr\}@email.arizona.edu} }

\maketitle \thispagestyle{empty}

\newcommand\blfootnote[1]{%
  \begingroup
  \renewcommand\thefootnote{}\footnote{#1}%
  \addtocounter{footnote}{-1}%
  \endgroup
}

\blfootnote{This work was supported by the NSF grant CAREER-1651492.}

\vspace{-28pt}
\begin{abstract}
Owing to data-intensive large-scale applications, distributed computation systems have gained significant recent interest, due to their ability of running such tasks over a large number of commodity nodes in a time efficient manner. One of the major bottlenecks that adversely impacts the time efficiency is the  \textit{computational heterogeneity} of distributed nodes, often limiting the task completion time due to the slowest worker. 
 In this paper, we first present a lower bound on the expected computation time based on the work-conservation principle.  We then present our approach of \textit{work exchange} to combat the latency problem, in which faster workers can be reassigned additional leftover computations that were originally assigned to slower workers. We present two variations of the work exchange approach: a) when the computational heterogeneity knowledge is known a priori; and b) when heterogeneity is unknown and is estimated in an online manner to assign tasks to distributed workers. As a baseline, we also present and analyze the use of an optimized Maximum Distance Separable (MDS) coded distributed computation scheme over heterogeneous nodes. Simulation results also compare the proposed approach of work exchange, the baseline MDS coded scheme and the lower bound obtained via work-conservation principle.  We show that the work exchange scheme achieves time for computation which is very close to the lower bound with limited coordination and communication overhead even when the knowledge about heterogeneity levels is not available.
\end{abstract}

%

\vspace{-18pt}
\section{Introduction}
\label{sec:Introduction}
\vspace{-3pt}

Due to the astonishing expansion of the data production in the recent years, processing of large scale data-sets over a centralized single machine is widely becoming infeasible. This growth moves the computational paradigm towards large distributed computation, which can enable the processing of data-intensive tasks for machine learning, and data mining over a large number of  commodity machines in a time-efficient manner (e.g., Apache Spark \cite{Apache2010}, and MapReduce \cite{MapReduce2004}).
One of the concerns related to large-scale distributed computation is due to \textit{stragglers}, which are a small set of machines/nodes that can unexpectedly run slower than the average or can be even disconnected due to several factors such as resource contention, network or disk failures, and power limits \cite{Stragglers1,Stragglers2}. 
Being buckled under faults, the task will be delayed or even never completed when the master node has to wait for extended periods pf time until the stragglers reply.

In order to mitigate the impact of stragglers, various approaches have been proposed from several research communities, as discussed next. 
One approach is to efficiently detect the stragglers while running computational tasks, and then relaunch the delayed tasks on other machines \cite{Stragglers2,Stragglers3}.
Another set of approaches attack this problem by introducing redundancy in computation. The simplest form of redundancy, i.e., repetition was proposed in \cite{Stragglers/replicas1,Stragglers/replicas2}, where each repeated task can combat one straggler doing the same task.
However, a more efficient approach to introduce redundancy is by leveraging codes. In particular, the idea of using $(n,k)$ Maximum Distance Separable (MDS) coded distributed computations was presented in \cite{CodedComp-Shuffle2017} for linear distributed computations such as matrix multiplication, where $n$ is the number of workers.
The computation time is then limited by the $k^{th}$ (instead of the $n^{th}$) fastest worker. The performance can also be quantified in terms of the \textit{recovery threshold} metric, where any $k$ out of $n$ workers are sufficient to complete the whole underlying task, and hence the system can tolerate up to any $(n-k)$ stragglers.
The result in \cite{CodedComp-Shuffle2017} was extended in  \cite{Stragglers-MDSproduct} for high dimensional matrix multiplication. More recently, a new class of codes (polynomial codes) were proposed in \cite{Stragglers-Polynomial}, which achieve the optimal recovery threshold for high dimensional matrix multiplication. Moreover, in \cite{CDC1}, and \cite{CDC2}, an efficient coded framework was proposed using MDS codes for distributed matrix multiplication in the MapReduce setting. 
In \cite{Stragglers-Conv}, the authors proposed the use of codes in parallelly computing the convolution of two long vectors before a deadline in the presence of stragglers.
In addition to coded computation, there are other other hardware based approaches approaches to combat computational heterogeneity by explicitly attempting to reduce the unevenness and variability related to hardware effects, e.g., \cite{Stragglers/HW}.
Another approach (with roots in operations research commpunity) is that of  \textit{work stealing}, where the faster workers  take over the remaining computations from the slower workers once they finish their tasks leaving the slower workers idle till the end of the computation session \cite{work-stealing}.

At the core of the straggler problem is the \textit{heterogeneity of computation across the workers}, i.e., different workers in the cluster may have different computational capabilities. The recently proposed approaches that employ codes to combat this problem have the drawback that they only use the computations from a subset of workers, and hence, the computations performed from the remaining nodes are not used at all. Furthermore, other approaches such as work stealing and variations suffer from the drawback that slower workers can remain idle for certain periods of time. In summary,  all of the aforementioned approaches suffer either from redundant computations or idle workers not participating for a period of time. 
 In this work, we explore a series of fundamental questions related to this problem: 

\begin{itemize}
\item \textit{$Q_1$: If no redundant computations are permitted, nor any worker remains idle at any point, what is the minimum time for distributed computation?} To answer this question, we present an orcale lower bound on the computation time which is a function of the underlying computational speeds of the workers (i.e., the parameters describing the computational heterogeneity).  

\item \textit{$Q_2$: How do we devise schemes that can approach the oracle lower bound?} To answer this question, we present a new idea called \textit{work exchange}, which attempts to minimize redundant computations, by reassigning and exchanging computational tasks across workers. This work exchange approach is fundamentally different from the coded computation approaches in two ways: a) it attempts to minimize redundant computation and leverage the knowledge of computational heterogeneity in the system; and b) it works for linear or non-linear computational tasks as long as the original task is sub-divisible into large number of smaller sub-tasks. 

\item \textit{$Q_3$: What are the resources that can be traded-off in order to minimize the computation time?}  Coded distributed computation approaches (in particular, the ones using MDS codes) require excess number of workers (in contrast to waiting for the slowest worker), however they do not require intermediate coordination once the tasks are assigned. On the other hand, the proposed work exchange approach leads to a reduction in the computation time, however, it requires more coordination/communication overhead for the re-assignment of tasks.   Thus, the fundamental question is that what are the tradeoffs between communication, storage and coordination in order to achieve a certain time for distributed computation.

\end{itemize}


To give an intuition behind leveraging the knowledge of computational heterogeneity, let us consider a matrix multiplication problem where we need to compute $Ax$, where $A\in\mathbb{R}^{200\times d}$ and $x\in \mathbb{R}^{d\times 1}$ over a cluster of three distributed workers. Clearly, this problem can be divided into $200d$ elementary multiplication operations. The three workers are heterogeneous and for simplicity, let us assume that they work at a constant rate  $\lambda_1=d$, $\lambda_2=3d$, and $\lambda_3=6d$ operations per second for workers 1, 2, and 3 respectively. 

\noindent \textit{\underline{Heterogeneity unaware MDS Coded Computation}:} The coding based approach works as follows: we first divide the matrix $A$ across rows into two equal sub-matrices $A^{(1)}_1$, and $A^{(1)}_2$ of size $100\times d$ each. Then, we use a $(3,2)$ MDS code over these sub-matrices to generate $3$ coded matrices, e.g., $A^{(1)}_1$, $A^{(1)}_2$, and $A^{(1)}_1+A^{(1)}_2$, to be assigned to the three workers as shown in Figure~\ref{fig:example}a. Every worker then multiplies the assigned coded matrix with $x$, which needs $100d$ operations for each worker. According to the speed of each worker, the master node receives the replies from workers 1, 2, and 3 after $100$, $33.33$, and $16.67$ seconds, respectively.
However, due to the MDS property, any 2 out of 3 replies are enough to decode $A^{(1)}_1x$, and $A^{(1)}_2 x$, yielding the completion time to be $33.33$ seconds.

\noindent \textit{\underline{Heterogeneity aware Distributed Computation}:} Consider now the following approach, where instead of adding redundancy, we break the matrix $A$ into three unequal sized sub-matrices $A^{(2)}_1$, $A^{(2)}_2$, and $A^{(2)}_3$ of sizes $20\times d$, $60\times d$, and $120\times d$, respectively, and then assign them to workers 1, 2, and 3, respectively as shown in Figure~\ref{fig:example}b.
According to the speed of each worker, the master node receives the  replies from the three workers simultaneously after $20$ seconds, which is less than $33.33$ seconds achieved using the coding approach. From this simple example, we make the following key observations: a) if we are aware of computational heterogeneity (or if it can be estimated), then we can leverage such prior heterogeneity knowledge in order to balance out the computational operations assigned to the workers;  b) ignoring slow workers or introducing redundant computations is not always optimal from a latency perspective.

\begin{figure}[t]
  \begin{center}
  \includegraphics[width=0.85\columnwidth]{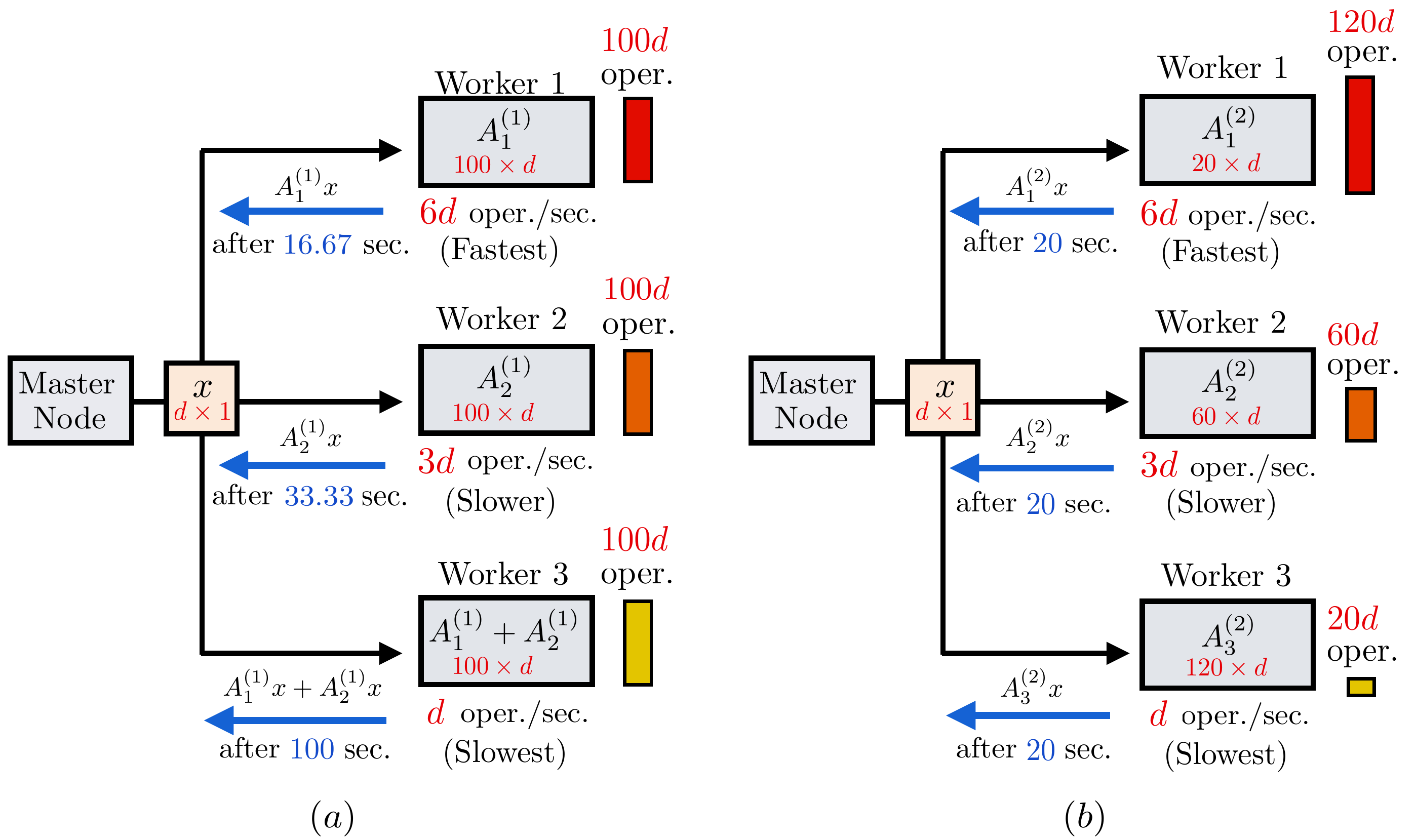}
  \vspace{-5pt}
\caption{ {The matrix multiplication $Ax$ problem over three workers with different computational speeds. In Figure 1(a), a $(3,2)$ MDS code is used so that any $2$  replies are enough to decode $Ax$. This coded computation scheme achieves a computation time of $33.33$ seconds. In Figure 1(b), the computation of $Ax$ is unevenly divided across the workers according to the heterogeneity knowledge so that all workers finish simultaneously after $20$ seconds.}}
\label{fig:example}
\vspace{-25pt}
  \end{center}
\end{figure}

\noindent We next summarize the main contributions of this paper:

\noindent \hspace{4pt}$\bullet$ We first study a baseline scheme which uses $(n,k)$ MDS codes for distributed computation, and  find the expected mean computation time optimized over all choices of $k \in [1:n]$.

\noindent \hspace{4pt}$\bullet$ We then study the work conservation principle, and derive an oracle lower bound on the computation time as a function of the computational heterogeneity of the system. We find that this oracle lower bound bears an interesting analogy to the well known \textit{water-filling} solution for optimal power allocation over parallel channels with heterogenous noise characteristics. In the optimal water-filling solution, more power is assigned to channels with higher signal-to-noise ratio (SNR). Similarly, the oracle lower bound naturally suggests that to minimize computation time, more computations must be assigned to workers with faster computational capabilities.

\noindent \hspace{4pt}$\bullet$ We then propose a work-exchange  based scheme, with the goal of approaching the oracle bound, through limited coordination and communication among workers. 
We first consider the case when the heterogeneity knowledge is known apriori. We then address a more realistic scenario where the heterogeneity knowledge is not available, and can be gradually learned  over time allowing for larger levels of coordination and communication as compared to the former case.

\noindent \hspace{4pt}$\bullet$ We present simulation results to compare the baseline optimized MDS coded scheme, the oracle lower bound and the variations of  work-exchange based scheme. Interestingly, the work-exchange scheme is shown to be very close to the lower bound and requires limited amount of coordination and communication even without heterogeneity knowledge.

\vspace{-7pt}
\section{System Model}
\label{sec:System}
\vspace{-7pt}

We consider a distributed computation system, where the master node has the entire data-set of $N$ data points.
The master node sends batches of the data-set to a set of $K$ distributed workers in order to locally calculate some function or train a model in a distributed manner. 
Each worker independently starts to compute a function of the assigned batch (as an example, this function could correspond to the gradient or sub-gradients of the data points assigned to the worker).
Then, the local functions are fed-back to the master node for further processing.
For this work, we assume the computations can be divided into smaller and independent sub-tasks (i.e., $N$ is large enough) and the task is not sequential which means that the results do not depend on the order in which the data is being processed, and the overall computational task terminates once the master node gets results for the total $N$ data points in any order.
The time utilized to process a single data point  is random, and is modeled at any worker $w_k$ as an exponential random variable with average speed $\lambda_k$ number of processed points per second, i.e., $t^n_k\sim \text{exp}(\lambda_k)$ and $t^n_k$ is the time used to process a point $n=\{1,2,\ldots, N\}$ at worker $w_k$, where $k\in \bar{K}=\{1,\ldots,K\}$. We also define $\bar{\lambda}=\{\lambda_1,\ldots, \lambda_K\}$ as the heterogeneity parameter set for the $K$ workers.

The master node initially assigns the data-set to the $K$ workers in data batches spanning the whole data-set.
Our system model also allows for reassigning data points to the workers according to some conditions in order to speed up the process, where the new assignments span the points that are not yet processed.
We define $N_{\text{assign}}^{(k,i)}$ as the number of points assigned to worker $w_k$ at the reassigning iteration $i\in\{1,2,\ldots, I\}$, where $I$ represents the total number of reassigning epochs/iterations. 
 We also  define $N_{\text{done}}^{(k,i)}$ as the number of points processed by worker $w_k$ during the $i^{th}$ iteration, and $N_{\text{done}}^{(k)}=\sum_{i=1}^I N_{\text{done}}^{(k,i)}$ as the total number of points processed by worker $w_k$ over the entire process.
%
The reassignment process imposes extra communication overhead, which is given for a worker $w_k$ and the $i^{th}$ assignment by 
\begin{align}
N_{\text{comm}}^{(k,i)}=\max(N_{\text{assign}}^{(k,i)}-N_{\text{left}}^{(k,i-1)},0),
\end{align}
where $N_{\text{left}}^{(k,i-1)}=N_{\text{assign}}^{(k,i-1)}-N_{\text{done}}^{(k,i-1)}$ is number of points not processed by $w_k$ from the previous assignment. The total extra communication overhead (not including the initial assignment) is therefore 
\begin{align}
\label{eq:comm-overhead}
N_{\text{comm}} = \sum _{k=1}^K \sum _{i=2}^I N_{\text{comm}}^{(k,i)}.
\end{align}
The end-to-end completion time $T$ can be written now as the sum of three components:
$T= T_{\text{comp}} +T_{\text{comm}} +T_{\text{corr}}$,
where $T_{\text{comp}}$ is the time for computation, $T_{\text{comm}}$ is the time for communication and is directly proportional to $N_{\text{comm}}$, and $T_{\text{corr}}$ is the time for coordination and is directly proportional to the number of reassignments $I$. In this work, we will study the terms \{$T_{\text{comp}}$, $N_{\text{comm}}$, $I$\}, and the trade-offs between them.

%

 \section{Baseline MDS Coded Computation Scheme}

We first consider the MDS coded computation scheme recently developed  in \cite{CodedComp-Shuffle2017} to combat stragglers as a baseline to be compared with the lower bound and the proposed work-exchange schemes. For fair comparison, we consider the same number of distributed workers $K$ to be used in order to process a data-set of $N$ data points, where MDS codes can be used as follows: a) the data-set is divided into equal chunks of $N/L$ points each, for some $ L\in \bar{K}=\{1,2,\ldots, K\}$; b) $(K,L)$ MDS codes can be used to assign coded chunks of data to the workers, such that  the learning process is done as soon as any $L$ out of the $K$ workers reply back with the processed points; and then c) the value of $L\in \bar{K}$ can be optimized in order to get the lowest possible computation time. 

Let us now define the random variable $T_{k,m}$ as the time required to process $m$ points by a worker $w_k$ with rate $\lambda_k$, then $T_{k,m}$ (which is a sum of exponential random variables) follows an Erlang distribution, i.e., $T_{k,m}\sim \text{Erlang}(m,\lambda_k)$.
Using $(K,L)$ MDS codes, the total computation time  is limited by the $L^{\text{th}}$ fastest worker instead of the slowest worker, i.e., $T_{\text{comp}}^{\text{MDS}}(L)= T^{(L,m)}$ at $m=N/L$, where $T^{(L,m)}$ is the $L^{\text{th}}$ ordered statistic for the random variables $T_{k,m}$, where $k \in \bar{K}$.
Therefore, the mean value for  $T_{\text{comp}}^{\text{MDS}}(L)$ is 
\begin{align}
E[T_{\text{comp}}^{\text{MDS}}(L)]= \mu_{(L,N/L)},
\end{align}
where $\mu_{(L,m)}$ is defined as the mean of the $L^{\text{th}}$ ordered statistic for $K$ independent and non-identically distributed Erlang random variables $T_{k,m},\,k\in\bar{K}$, with rate $\lambda_k$ and a shape parameter $m$. In order to find $\mu_{(L,m)}$, we use the results in \cite{mean-os-erlang2004} and apply them to our model after some mathematical manipulations (omitted here due to space constraints), where using the following recursion relation we can find $\mu_{(L,m)}$ starting from $\mu_{(0,m)}=0$ as follows:
\begin{align}
\label{eq:mean_m_K}
\hspace{-5pt}\mu_{(\ell,m)} =\mu_{(\ell-1,m)} +\sum_{j=1}^{\ell} (-1)^{j-1}\hspace{-3pt}
\left(\hspace{-5pt}\begin{array}{c}
K-\ell+j\\ j-1
\end{array}\hspace{-5pt}\right)
P_{K-\ell+j}^m,
\end{align}
where $\ell\in\{1,2,\ldots,L\}$, and $P_{j}^m$ can be written as
\begin{align}
\label{eq:P_baseline}
 P_{j}^m= \sum_{\mathcal{K}_j\subseteq \bar{K}} \sum_{\substack{0\leq n_i<m\\i=1,\ldots,j}}\frac{1}{ \lambda_{\mathcal{K}_j}}
\left(\frac{(\sum_{i} n_i)\,!}{n_{1}!\ldots\,n_{j}!}\right)
\prod_{i=1}^j \left(\frac{\lambda_{ki}}{\lambda_{\mathcal{K}_j}}\right)^{n_{i}},
\end{align}
where $\mathcal{K}_j=\{k_1,\ldots,k_j\}\subseteq \bar{K}$ is an arbitrary sub-set of $j$ workers, and $\lambda_{\mathcal{K}_j}=\sum_{i=1}^j \lambda_{k_i}$.
Finding $\mu_{L,m}$ using the recursion relation in (\ref{eq:mean_m_K}), $\forall L\in\bar{K}$, we then can get the optimized mean of the MDS coded computation scheme as follows
\begin{align}
\label{eq:mean_MDS}
E[T_{\text{comp}}^{\text{MDS}}] = \min_{L\in \bar{K}} \mu_{(L,N/L)}.
\end{align}

\section{Lower Bound (Work Conservation Principle)}
In this section, we present an oracle lower bound on the computation time using the work conservation principle. 
For the oracle bound analysis, we assume a genie aided system satisfying the work conservation principle with the following assumptions:
\begin{enumerate}
\item  All the $K$ workers have the entire data-set available in their storage;
\item Complete coordination between the workers is permitted and the data points computed at each worker have no overlap (i.e., each worker performs completely distinct computations), i.e., $\sum_{i=1}^K N_{\text{done}}^{(k)}=N$; 
\item No worker remains idle at any point of the learning process.
\end{enumerate}

We now define $T_{\text{comp}}^{\text{oracle}}$ as a random variable representing the total time needed to finish the computations for all the $N$ data points across the $K$ workers. We further define the random variable $N_{T_{\text{comp}}^{\text{oracle}}}(t)$ as the total number of points processed across the $K$ workers at time $t$. The following Theorem provides an oracle bound for the average computation time, as well as the expected number of data points  processed by every worker towards completion following the work conservation principle.
\begin{theorem}
\label{thm1}
The average computation time for the oracle scheme on a set of $K$ distributed workers of heterogeneity set $\bar{\lambda}$, and collaborating on the processing of $N$ number of data points is given by
\begin{align}
\label{eq:thm1}
E[T_{\text{comp}}^{\text{oracle}}] =\frac{N}{\lambda_{\text{sum}}}, \text{~~~where $\lambda_{\text{sum}}=\sum_{i=1}^K \lambda_k$.}
\end{align}
\end{theorem}

This lower bound can be viewed as the expected computation time for an oracle machine working with the sum of the speeds ($\lambda_{\text{sum}}$) of the distributed workers.

\begin{proof}The CDF of $T_{\text{comp}}^{\text{oracle}}$ for $t\geq 0$ can be found as
\begin{align}
F_{T_{\text{comp}}^{\text{oracle}}}(t) &= \Pr(T_{\text{comp}}^{\text{oracle}}\leq t ) =\Pr(N_{T_{\text{comp}}^{\text{oracle}}}(t)\geq N)\nonumber\\
&=\sum_{\substack{n_k\in \bar{n}\\n_{\text{sum}}=N}} \prod_{k=1}^K \Pr \left(N_{\text{done}}^{(k)}(t)\geq n_k\right)\nonumber\\
&=\sum_{\substack{n_k\in \bar{n}\\n_{\text{sum}}\geq N}} \text{P}_{\bar{n}}(t)\nonumber\\
&=1-\sum_{\substack{n_k\in \bar{n}\\n_{\text{sum}}<N}} \text{P}_{\bar{n}}(t),
\end{align}
 where $N_{T_{\text{comp}}^{\text{oracle}}}(t)$ is the total number of points processed at time $t$, $n_{\text{sum}}=\sum_{k=1}^K n_k$, $\bar{n}=\{n_1,\ldots,n_K\}$,  and $\text{P}_{\bar{n}}(t)$ is given as follows
\begin{align}
\label{eq:P_n}
\text{P}_{\bar{n}}(t)&=\prod_{k=1}^K \Pr \left(N_{\text{done}}^{(k)}(t)= n_k\right)= \prod_{k=1}^K \frac{e^{-\lambda_k t} (\lambda_k t)^{n_k}}{n_k!}.
\end{align}

Since $T_{\text{comp}}^{\text{oracle}}$ is a non-negative random variable, the average time $E[T_{\text{comp}}^{\text{oracle}}]$ to finish all the $N$ point across the $K$ workers can be found as
\begin{align}
\label{eq:mean_oracle1}
\hspace{-2pt}E[T_{\text{comp}}^{\text{oracle}}] \hspace{-2pt}=\hspace{-2pt} \int_{0}^{\infty}\hspace{-3pt}\left(1-F_{T_{\text{comp}}^{\text{oracle}}}(t)\right) dt\hspace{-2pt}=\hspace{-2pt}\sum_{\substack{n_k\in \bar{n}\\n_{\text{sum}}<N}} \hspace{-2pt}\int_{0}^{\infty} \text{P}_{\bar{n}}(t) dt,
\end{align}
where the integration in (\ref{eq:mean_oracle1}) can be found as
\begin{align}
\label{eq:integration_oracle}
\hspace{-3pt}\int_{0}^{\infty}\hspace{-2pt} \text{P}_{\bar{n}}(t)dt= 
\frac{1}{\lambda_{\text{sum}}}
\left(\frac{n_{\bar{K}}\,!}{n_{1}!\ldots\,n_{K}!}\right)
\prod_{k=1}^K \left(\frac{\lambda_k}{\lambda_{\text{sum}}}\right)^{n_k}.
\end{align}
\noindent Using (\ref{eq:integration_oracle}), we can now rewrite the summation in (\ref{eq:mean_oracle1}) as
\begin{align}
\label{eq:mean_oracle2}
E[T_{\text{comp}}^{\text{oracle}}] &= \frac{1}{\lambda_{\text{sum}}} \sum_{n=0}^{N-1} \sum_{\substack{n_k\in \bar{n}\\n_{\text{sum}}=n}}\left(\begin{array}{c}
n \\ n_1,\ldots,n_K
\end{array}\right)
\prod_{k=1}^K \left(\frac{\lambda_k}{\lambda_{\text{sum}}}\right)^{n_k}\nonumber\\
&=\frac{1}{\lambda_{\text{sum}}} \sum_{n=0}^{N-1} 1= \frac{N}{\lambda_{\text{sum}}},
\end{align}
where the second equality follows from the Multinomial theorem, which proves Theorem~\ref{thm1}.
\end{proof}

\begin{corollary}
\label{coro1}
The expected total number of points processed at completion by worker $w_k$ with mean speed $\lambda_k$ is given by
\begin{align}
\label{eq:coro1}
E[N_{\text{done}}^{(k)}]= \lambda_k t = \frac{\lambda_k N}{\lambda_{\text{sum}}}.
\end{align}
\end{corollary}
\begin{proof}
Let us fix a time frame  $t=\frac{N}{\lambda_{\text{sum}}}$, which is the mean completion time. The number of points done by a worker $w_k$ in time $t$ is a Poisson random variable with mean $\lambda_k t$, i.e., $N_{\text{done}}^{(k)}\sim \text{Poisson}(\lambda_k t= \frac{\lambda_k N}{\lambda_{\text{sum}}})$. Therefore, the mean number of points processed by $w_k$ by completion is given as in (\ref{eq:coro1}).
\end{proof}

\section{Heterogeneity Aware Work-Exchange}

In this section, we try to approach the oracle bound in (\ref{eq:thm1}) using the prior heterogeneity knowledge of the workers reflected by the values of $\lambda_k$, for $k=1,\ldots, K$. 
According to  (\ref{eq:coro1}), the average number of data points processed by a worker $w_k$ is directly proportional to $\lambda_k$.
That is, we need to assign smaller tasks for slower workers, and larger tasks for faster workers. 
In our first algorithm, we we present a heterogeneity aware scheme in which do not allow any intermediate coordination or communication between the workers and the master once the data is initially assigned. That is, the tasks assignment for the workers remains fixed throughout the computation process.  We next introduce the work exchange approaches by allowing for a low level of intermediate coordination and communication between the workers and the master in order to further reduce the latency.

\subsection{Heterogeneity Aware Scheme with Fixed Work Assignment}\label{SecAlgo1}

The basic idea of this scheme is that the master node assigns different and non-overlapping portions of data for every worker depending on the heterogeneity knowledge such that the worker with higher $\lambda$ gets more data assignment as follows
\begin{align}
\label{eq:assign-NC}
N_{\text{assign}}^{(k)}= \frac{\lambda_k N}{\lambda_{\text{sum}}}, \quad \forall k\in\bar{K}.
\end{align}

\begin{figure}[t]
  \begin{center}
  \includegraphics[width=0.6\columnwidth]{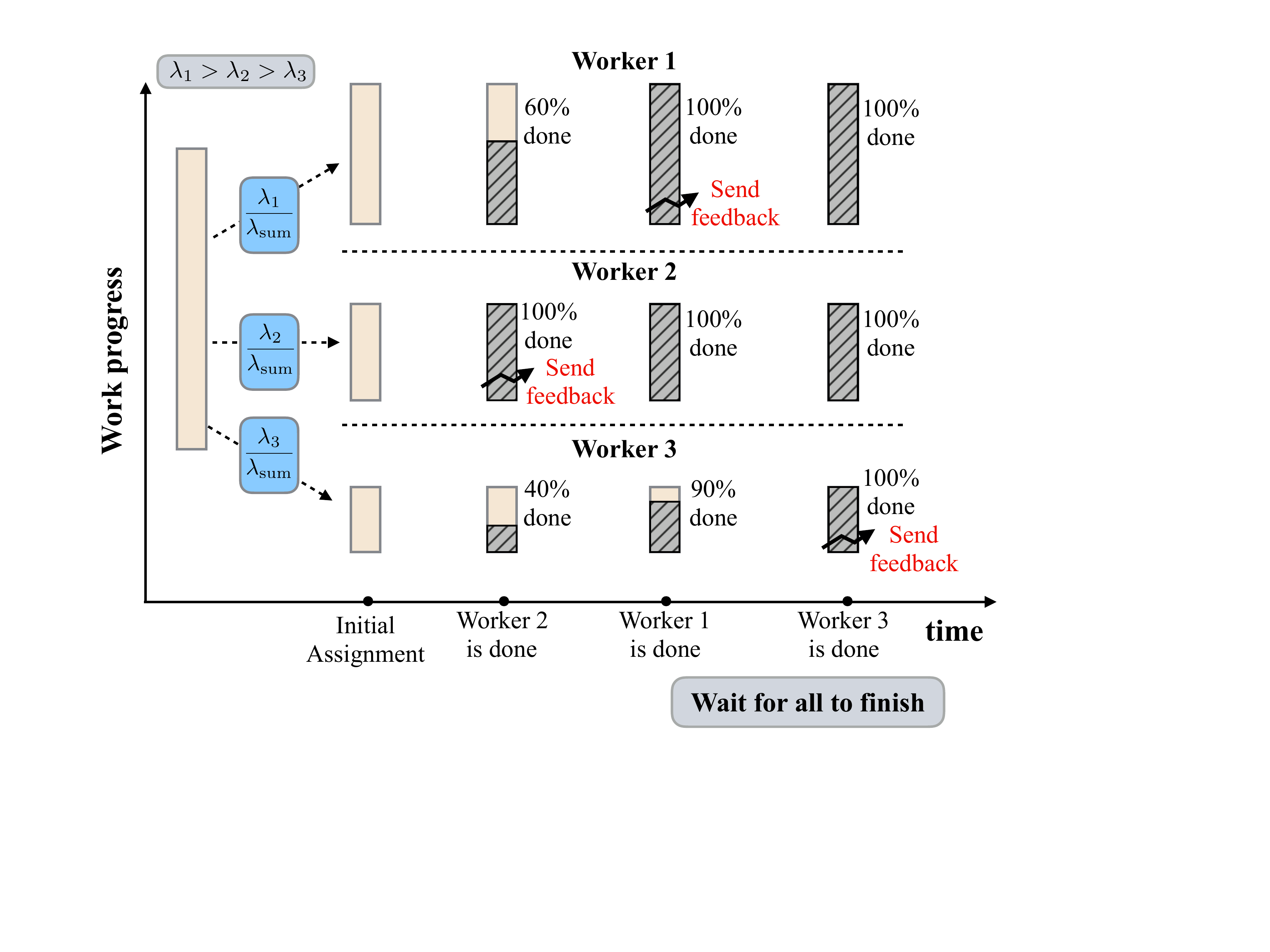}
\caption{{(Fixed work assignment Algorithm in Section  \ref{SecAlgo1}) The work is initially assigned according to the heterogeneity knowledge such that faster workers get larger assignments and vice versa. The master node waits for all the workers to reply.}}
\label{fig:sch1}
\vspace{-20pt}
  \end{center}
\end{figure}

The master node waits for all the workers to reply back as shown in Figure~\ref{fig:sch1}. 
That is the computation time is limited by the last worker to reply.
The completion time for worker $w_k$ is the time needed to process $N_{\text{assign}}^{(k)}$ data-points at a rate $\lambda_k$, which follows an Erlang distribution, i.e.,
$T_{k}\sim \text{Erlang}(N_{\text{assign}}^{(k)},\lambda_k)$.
Therefore, the total computation time is
\begin{align}
\label{eq:heter-aware}
T_{\text{comp}}=\max(T_{1},T_{2},\ldots,T_{K}).
\end{align}

\subsection{Work Exchange with Heterogeneity Knowledge}\label{SecAlgo2}

We next present the work exchange scheme for the scenario in which the master node has knowledge about computational heterogeneity. 
Similar to the previous scheme, the initial assignment depends on the heterogeneity knowledge as follows
\begin{align}
\label{eq:assign-WC}
N_{\text{assign}}^{(k,1)}= \frac{\lambda_k N}{\lambda_{\text{sum}}}, \quad \forall k\in\bar{K}.
\end{align}
The completion time for worker $w_k$ in the first iteration, $T_{k}^{(1)}$, is the time needed to process $N_{\text{assign}}^{(k,1)}$ data-points at a rate $\lambda_k$, which follows an Erlang distribution, i.e.,
$T_{k}^{(1)}\sim \text{Erlang}(N_{\text{assign}}^{(k,1)},\lambda_k)$.
Once any worker $w_k$ finishes the processing over its data chunk it feeds-back the master-node, and declares that it is now idle by sending a completion flag $f^{(k)}=1$ as shown in Figure~\ref{fig:sch2}. 
The time for the first iteration then is given by time when the first worker replies, i.e., 
\begin{align}
T_{\text{comp}}^{(1)}=\min(T_{1}^{(1)},T_{2}^{(1)},\ldots,T_{K}^{(1)}).
\end{align}

\begin{figure}[t]
  \begin{center}
  \includegraphics[width=0.75\columnwidth]{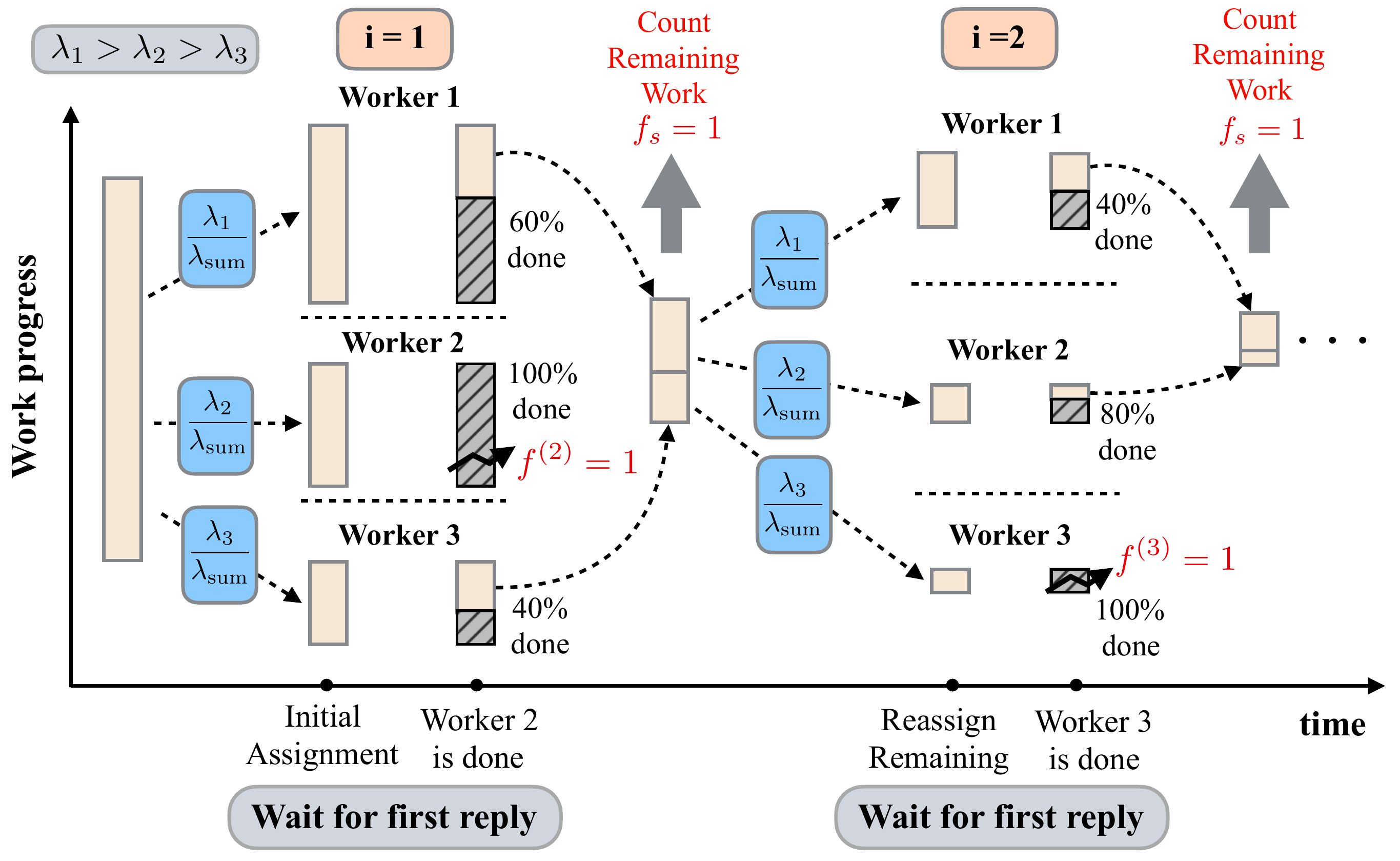}
\caption{{(Heterogeneity Aware Work Exchange Algorithm in Section  \ref{SecAlgo2}) The work is initially assigned according to the heterogeneity knowledge. The master node waits for the first reply, and then performs work exchange reassignments. The process continues over the remaining work till all the computations are collected.}}
\label{fig:sch2}
\vspace{-20pt}
  \end{center}
\end{figure}

The master-node then pauses the processing over all workers by broadcasting a stop flag $f^s=1$, and gets \{$y^{(k,1)}$, $N_{\text{left}}^{(k,1)}$\}, $\forall k\in\bar{K}$, where for each worker $w_k$: $y^{(k,1)}$ is the latest work done before receiving $f^s=1$ after the initial reassigning iterations $i=1$, and $N_{\text{left}}^{(k,1)}$ is the number of points not processed yet.
The total number of points left after the initial reassigning iterations $i=1$ is then $N_{\text{rem}}^{(1)}=\sum_{k=1}^K N_{\text{left}}^{(k,1)}$. 
 The reassigning process in the next iteration is done as follows: the master node reassigns new non-overlapping portions of the remaining $N_{\text{rem}^{(1)}}$ points depending on the values of $\lambda$ for each worker using the same principle as before according to (\ref{eq:assign-WC}).
 
 Generally, the new assignment at  iteration $i$ is:
\begin{align}
\label{eq:assign_rem}
N_{\text{assign}}^{(k,i)}= \frac{\lambda_k N_{\text{rem}}^{(i-1)}}{\lambda_{\text{sum}}}, \quad \forall k\in\bar{K}.
\end{align}
The time for iteration $i$ is given by
$T_{\text{comp}}^{(i)}=\min(T_{1}^{(i)},T_{2}^{(i)},\ldots,T_{K}^{(i)})$,
where $T_{k}^{(i)}\sim \text{Erlang}(N_{\text{assign}}^{(k,i)},\lambda_k)$.
 The procedure can then be repeated for $I$ iterations until there are no remaining points left, $N_{\text{rem}}^{(I)} =0$.
 The total computation time is given by the sum of the computation times over all iterations
$T_{\text{comp}} = \sum_{i=1}^I T_{\text{comp}}^{(i)}$.
  The reassigning process imposes extra coordination overhead to get feedback from the workers at each reassigning epoch, and it also involves extra communication overhead since the master node needs to retransmit the new assigned points to the workers given by (\ref{eq:comm-overhead}).
  However, since every worker is initially assigned the expected number of points it can process, then for large number of points we expect negligible communication, and low coordination overheads. The complete steps of the work exchange scheme with heterogeneity knowledge are shown in Algorithm~\ref{alg1} for the master's node protocol and Algorithm~\ref{alg2} for the worker's protocol.

\begin{algorithm}
    \caption{Work exchange scheme with heterogeneity knowledge: Master node's protocol\label{alg1}}
  \begin{algorithmic}
    \INIT $N_{\text{rem}}= N$, $N_{\text{left}}^{(k)}=0$, $\forall k\in\bar{K}$
    \WHILE{$N_{\text{rem}} > N_{\text{rem}}^{\text{th}}$}
        \STATE Assign $N_{\text{assign}}^{(k)}=\frac{\lambda_k N_{\text{rem}}}{\lambda_{\text{sum}}}$, $\forall k\in\bar{K}$ \COMMENT{points assigned to each worker}
        \STATE Send $\max(N_{\text{assign}}^{(k)}-N_{\text{left}}^{(k)},0)$ new data points to each worker $w_k$
        \STATE Broadcast $f^s=0$
   	 	\STATE \textbf{on} receiving $f^{(k)}=1$, for any $k\in\bar{K}$
   	 	\STATE \quad Broadcast $f^s=1$
   	 	\STATE \quad Get \{$y^{(k)}$, $N_{\text{left}}^{(k)}$\}, $\forall k\in\bar{K}$
   	 	\COMMENT{get feedback from each worker}
   		 \STATE \quad Find $N_{\text{rem}}=\sum_{k\in\bar{K}} N_{\text{left}}^{(k)}$
    \ENDWHILE
    \STATE
  \end{algorithmic}
\end{algorithm}
\begin{algorithm}
    \caption{Worker node $w_k$'s protocol\label{alg2}}
  \begin{algorithmic}
    \STATE \textbf{On} receiving $N_{\text{assign}}^{(k)}$ points
 	\STATE Start processing
    \WHILE{$f^s==0$ \AND$\:f^{(k)}==0$}
    \STATE \quad \textbf{on} completing processing \textbf{set} $f^{(k)}=1$
    \ENDWHILE
    \STATE Feedback \{$y^{(k)}$, $N_{\text{left}}^{(k)}$, $f^{(k)}$\}
  \end{algorithmic}
\end{algorithm}

\begin{remark}[Cutting Threshold]
\label{rem1}
We observe that as the value of $N_{\text{rem}}$ decreases over iterations, smaller number of points are being assigned to workers. This may causing extra number of reassignment iterations. Instead, we add a cutting threshold $N_{\text{rem}}^{(I)} \leq N^{\text{th}}_{\text{rem}}$, where no more reassignment iterations are allowed, tolerating negligible delays by waiting for the slowest workers but over small number of remaining points $N^{\text{th}}_{\text{rem}}$. We note that the cutting threshold $N^{\text{th}}_{\text{rem}}$ is a design parameter and its choice impacts the overall completion time. The impact of the choice of this threshold is discussed through simulation results in Section \ref{Sec:Simulation}.
 \end{remark}

\section{Heterogeneity Unaware Work-exchange}\label{SecAlgo3}
We next introduce  a variation of the work exchange scheme, which assumes a more realistic scenario as follows:
a) the heterogeneity set, i.e., $\bar{\lambda}$, is  unknown apriori; and
b) each worker is not allowed any extra storage above $\frac{N}{K}$ points.
The steps for the scheme are depicted in Figure~\ref{fig:sch3}. As the master node initially does not know which worker is faster, it initially assigns equal data batches for the workers, i.e., $N_{\text{assign}}^{(k,1)}=\frac{N}{K}$ for $k\in\bar{K}$.
We can derive an estimate of the expected extra communication overhead for this scheme as the flow needed to restore the expected number of points to be processed by each worker in (\ref{eq:coro1}) starting from the uniform initial assignment, which is given by the following
\begin{align}
\label{eq:expected-comm-sch2}
E[N_{\text{comm}}] \approx \sum_{i=1}^K \max\left(\frac{N}{K}-\frac{N\lambda_k}{\lambda_{\text{sum}}},0\right).
\end{align}
\vspace{-10pt}

\begin{figure}[t]
  \begin{center}
  \includegraphics[width=0.75\columnwidth]{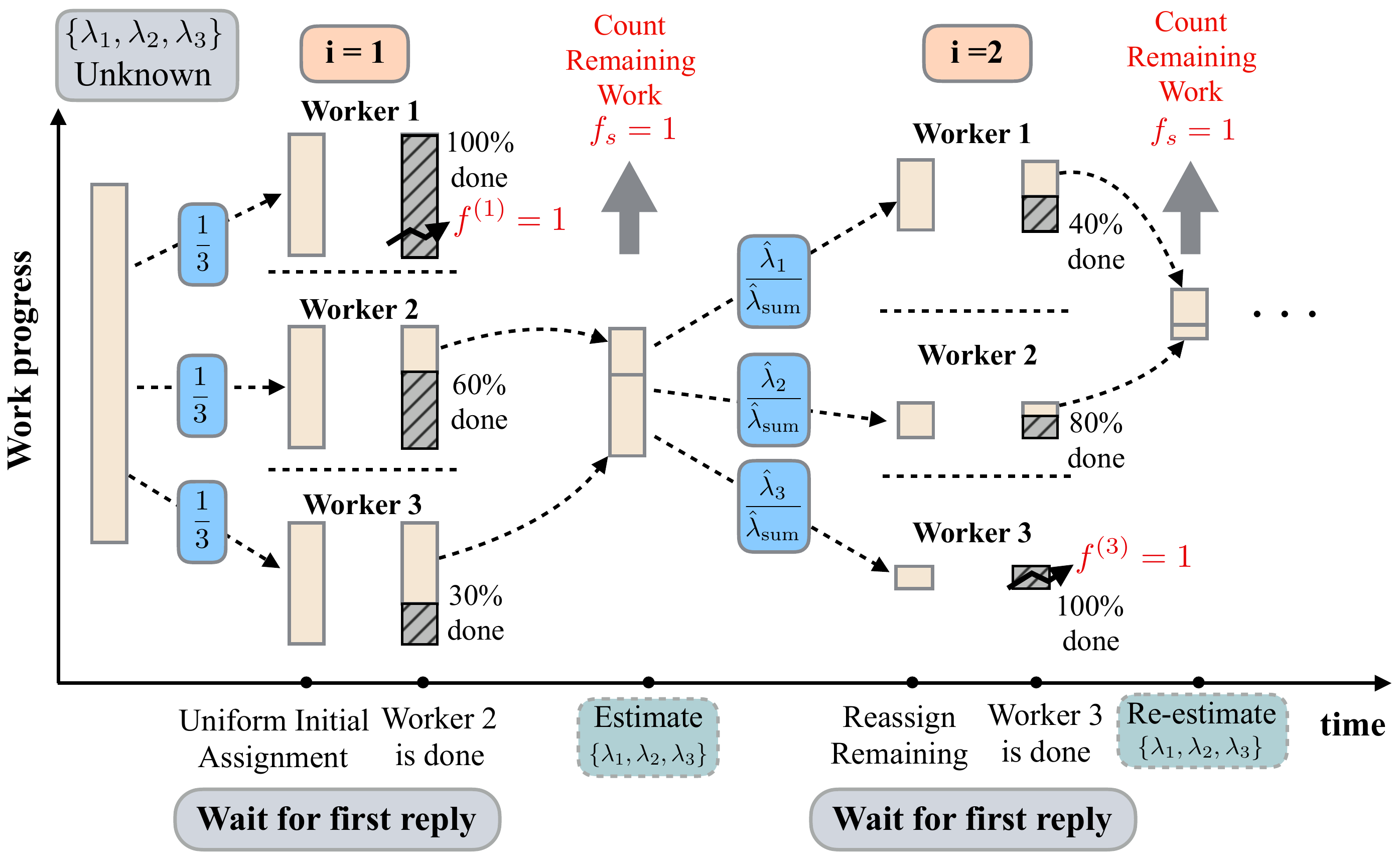}
\caption{{(Heterogeneity Unaware Work Exchange Algorithm in Section  \ref{SecAlgo3}) The work is uniformly assigned at the beginning. The master node waits for the first reply, estimates the average speed of each worker, and then reassigns the remaining work to all the workers according to the estimated heterogeneity. The process then continues till all the computations are completed.}}
\label{fig:sch3}
\vspace{-20pt}
  \end{center}
\end{figure}

Similar to the known heterogeneity case, upon receiving a completion flag $f^{(k)}=1$ from any of the workers, the master node stops the processing over all the workers by sending $f^s=1$, and gets the total number of points processed by each worker $w_k$ as $N_{\text{done}}^{(k,1)}$, and hence the total remaining points $N_{\text{rem}}^{(1)} =N - \sum_{k=1}^K N_{\text{done}}^{(k,1)}$.
The computation time for the first iteration is given by
\begin{align}
T_{\text{comp}}^{(1)}=\min(T_{1}^{(1)},T_{2}^{(1)},\ldots,T_{K}^{(1)}),
\end{align}
where $T_{k}^{(1)}\sim \text{Erlang}(N/K,\lambda_k)$. Before the next reassignment, the master node estimates the heterogeneity level as 
\begin{align}
  \hat{\lambda}_k^{(1)}=\frac{ N_{\text{done}}^{(k,1)}}{T_{\text{comp}}^{(1)}},\quad\forall k\in\bar{K}.
\end{align}
Therefore, the new reassignments can be given as
 \begin{align}
 \label{eq:reassign_sch2}
 N_{\text{assign}}^{(k,2)}= \min\left(\frac{N}{K},\frac{\hat{\lambda}_k^{(1)} N_{\text{rem}}^{(1)}}{\hat{\lambda}^{(1)}_{\text{sum}}}\right),
 \end{align}
 where $\hat{\lambda}^{(1)}_{\text{sum}} =\sum_{k=1}^K \hat{\lambda}^{(1)}_k$, and we take the minimum over those two values in order not to exceed the amount of storage $\frac{N}{K}$ allocated to each worker. 
 Generally for the next iterations, the heterogeneity is estimated as
 \begin{align}
  \hat{\lambda}_k^{(i-1)}=\frac{ \sum_{j=1}^{i-1}N_{\text{done}}^{(k,j)}}{\sum_{j=1}^{i-1}T_{\text{comp}}^{(j)}},\quad\forall k\in\bar{K},
\end{align}
while the new reassignments can be given as
 \begin{align}
 \label{eq:reassign_sch2}
 N_{\text{assign}}^{(k,i)}= \min\left(\frac{N}{K},\frac{\hat{\lambda}_k^{(i-1)} N_{\text{rem}}^{(i-1)}}{\hat{\lambda}^{(i-1)}_{\text{sum}}}\right),
 \end{align}
 where $N_{\text{rem}}^{(i-1)}= N- \sum_{j=1}^{i-1}\sum_{k=1}^{K}N_{\text{done}}^{(k,j)}$.
 It is clear that according to (\ref{eq:reassign_sch2}), that the 
sum over all $ N_{\text{assign}}^{(k,i+1)}$ for $k\in\bar{K}$ may not add necessarily to $N_\text{rem}^{(i)}$, and we simply carry on the difference for the next iteration.
Algorithm~\ref{alg3} summarizes the master's node protocol for the work exchange scheme with unknown heterogeneity, and the worker node's protocol works exactly as in Algorithm~\ref{alg2}. 
Furthermore, we can also add a cutting threshold $N_{\text{rem}}^{(I)} \leq N^{\text{th}}_{\text{rem}}$ over the remaining number of data points according to Remark~\ref{rem1}, to stop the reassigning process and reduce the level of coordination.

\begin{algorithm}
    \caption{Work exchange scheme with unknown heterogeneity: Master node's protocol\label{alg3}}
  \begin{algorithmic}
    \INIT $N_{\text{rem}}= N$, $\hat{\lambda}^{(k)}=1$, $N_{\text{left}}^{(k)}=0$, $N_{\text{done}}^{(k)}=0$, $\forall k\in\bar{K}$
    \WHILE{$N_{\text{rem}} > N_{\text{rem}}^{\text{th}}$}
        \STATE Assign $N_{\text{assign}}^{(k)}=\min\left(\frac{N}{K},\frac{\hat{\lambda}_k N_{\text{rem}}}{\hat{\lambda}_{\text{sum}}}\right)$, $\forall k\in\bar{K}$ \COMMENT{points assigned to each worker}
        \STATE Find $N_{\text{rem}}=N_{\text{rem}}-\sum N_{\text{assign}}^{(k)}$\COMMENT{Points to be carried on}
        \STATE Send $\max(N_{\text{assign}}^{(k)}-N_{\text{left}}^{(k)},0)$ points to each worker $w_k$
        \STATE Broadcast $f^s=0$
   	 	\STATE \textbf{on} receiving $f^{(k)}=1$, for any $k\in\bar{K}$
   	 	          \STATE \quad Broadcast $f^s=1$
   	 	\STATE \quad Get \{$y^{(k)}$, $N_{\text{left}}^{(k)}$\}, $\forall k\in\bar{K}$
   	 	\COMMENT{get feedback from each worker}
   		 \STATE \quad Find $N_{\text{rem}}=N_{\text{rem}}+\sum_{k\in\bar{K}} N_{\text{left}}^{(k)}$
   		  \STATE \quad Find $N_{\text{done}}^{(k)}=N_{\text{done}}^{(k)} +\left(N_{\text{assign}}^{(k)}- N_{\text{left}}^{(k)}\right)$, $\forall k\in\bar{K}$
   		  \STATE \quad Estimate $\hat{\lambda}_k=N_{\text{done}}^{(k)}/T_{\text{current}}$, $\forall k\in\bar{K}$
    \ENDWHILE
    \STATE
  \end{algorithmic}
\end{algorithm}

\section{Simulation Results}\label{Sec:Simulation}
In this section, we present simulation results to compare time for computation $T_\text{comp}$, the extra communication overhead $N_{\text{comm}}$, and the level of coordination needed (number of reassignments $I$) for the two proposed schemes against the oracle lower bound and the optimized baseline MDS coded scheme. In order to model the heterogeneity, we pick $\lambda_k$ for each worker uniformly at random  with mean $\mu$ and variance $\sigma^2$, i.e., $\lambda_k\sim \text{Uniform}(\mu-\sqrt{3\sigma^2},\mu+\sqrt{3\sigma^2})$, $\forall k\in\bar{K}$, therefore, for $\lambda_k\geq 0$, we have $0\leq\sigma^2\leq \mu^2/3$. The variance $\sigma^2$ indicates the level of heterogeneity, i.e., higher $\sigma^{2}$ leads to more heterogeneous workers, and when $\sigma^2=0$ then all the workers are homogeneous with $\lambda_k=\mu$, $\forall k \in \bar{K}$. In the following analysis, we run our proposed schemes for $N=10^6$ data points over $K=50$ distributed workers. The value of the cutting threshold (see Remark~\ref{rem1}) is given as a fraction of $\frac{N}{K}$, and is given the value $0.01 \frac{N}{K}$ by default.

\begin{figure}[t]
  \begin{center}
  \includegraphics[width=\columnwidth]{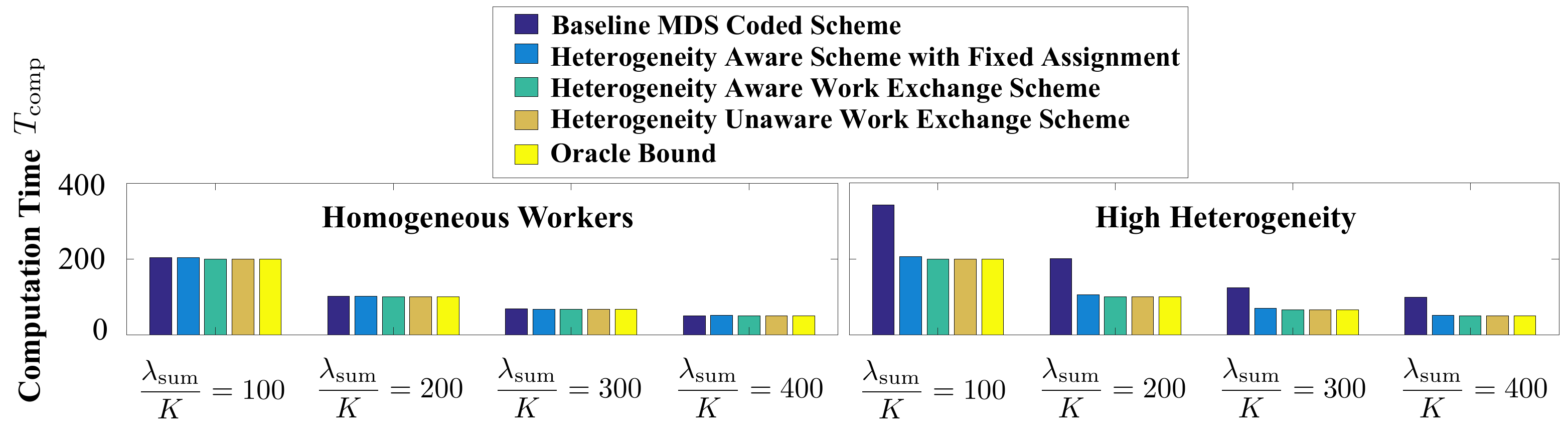}
\caption{{Computation time of $N=10^6$ data points over $K=50$ workers for different values of $\lambda_{\text{sum}}$, and two heterogeneity levels $\sigma^2 = 0$ and $\sigma^2 = \lambda^2_{\text{sum}}/(6K^2)$}.}
\label{fig:num1}
\vspace{-20pt}
  \end{center}
\end{figure}

\vspace{10pt}
\noindent\textbf{Comparison between the schemes:}
In Figure~\ref{fig:num1}, we find the mean computation time $T_{\text{comp}}$ for the work exchange schemes with and without heterogeneity knowledge, and compare them with 
the  optimized MDS coded scheme, the oracle bound,  and the
 heterogeneity aware scheme with fixed assignments  given in (\ref{eq:mean_MDS}), (\ref{eq:thm1}), and (\ref{eq:heter-aware}), respectively. We run the simulation for four different values of $\hat{\mu}=\lambda_{\text{sum}}/K$ (sample mean), and two heterogeneity levels: 1) $\sigma^2=0$, the homogeneous case; and 2) $\sigma^2=\hat{\mu}^2/6$, high heterogeneity.
We first notice, that both the work exchange schemes always have almost the same computation as the oracle bound $T_{\text{comp}}^{\text{oracle}}$, and a slightly higher values for the fixed assignment scheme. Second, as $\hat{\mu}$ increases, the value of $T_{\text{comp}}$ decreases for all the schemes as expected because of higher average speed for the workers. Most importantly, that $T_{\text{comp}}$ for the work exchange schemes are invariant of $\sigma^2$, unlike the value of $T_{\text{comp}}^{\text{MDS}}$ which equals $T_{\text{comp}}^{\text{oracle}}$ for the homogeneous case, and increases for higher value of $\sigma^2$. In fact, the redundancy introduced using MDS codes is unnecessary for the homogeneous case since all the workers perform the same, i.e., $L=K=50$ is the optimal value.

\begin{figure}[h!]
\nocaption
    \centering
    \begin{subfigure}[b]{0.48\textwidth}
        \includegraphics[width=\textwidth]{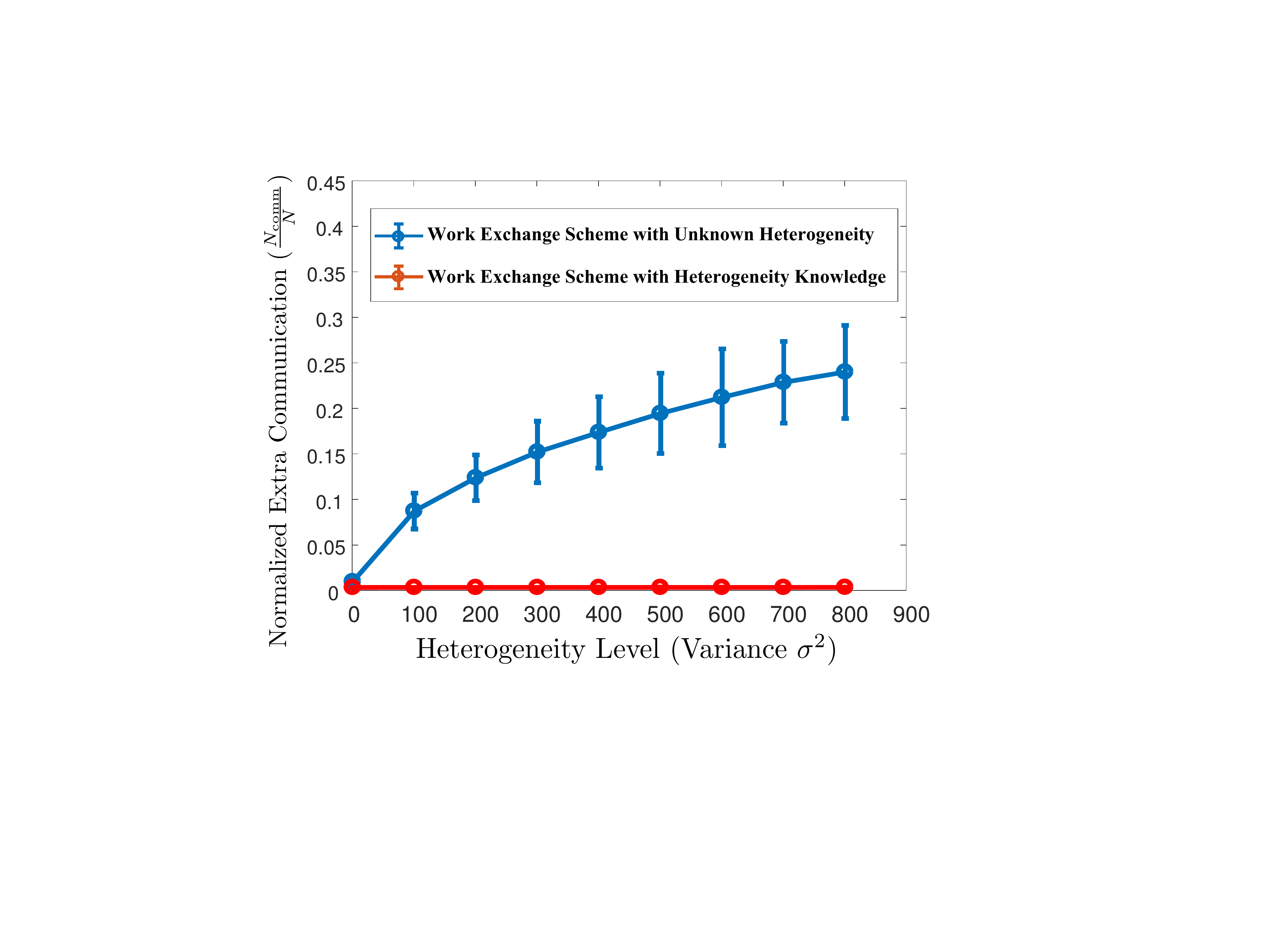}
\caption{Normalized extra communication  versus variance $(\sigma^2)$ for the work exchange scheme with and without the heterogeneity knowledge.\label{fig:num2}}
    \end{subfigure}
    ~
    \begin{subfigure}[b]{0.48\textwidth}
        \includegraphics[width=\textwidth]{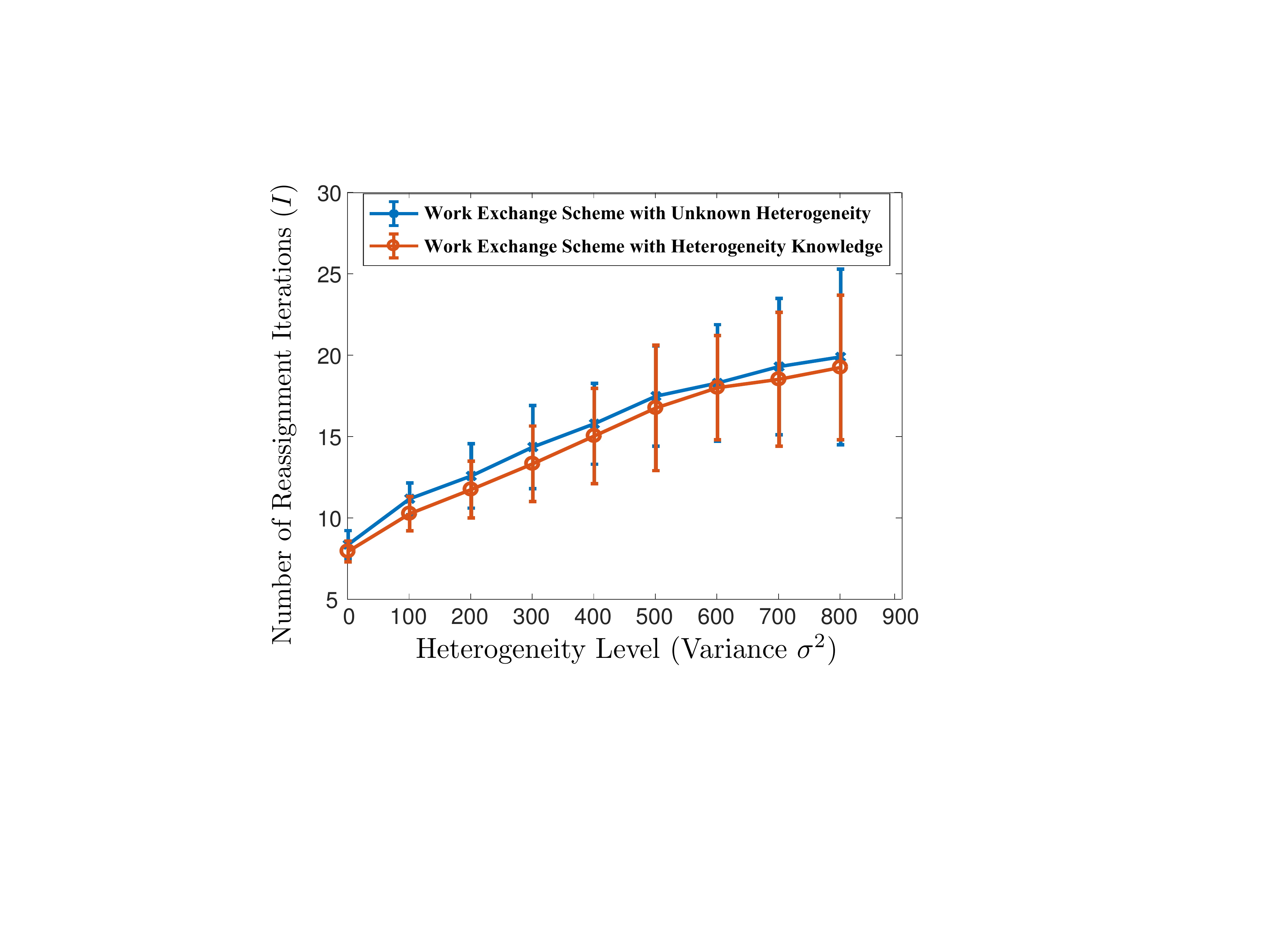}
        \caption{Number of reassignment iterations $(I)$ needed  versus variance $(\sigma^2)$ for the work exchange scheme with and without the heterogeneity knowledge.\label{fig:num3}}
    \end{subfigure}
    \vspace{-10pt}
\end{figure}

%

\vspace{10pt}
\noindent\textbf{Effect of heterogeneity level on  $N_{\text{comm}}$ and $I$:}
In Figures~\ref{fig:num2} and \ref{fig:num3}, we plot the average normalized extra communication $\frac{N_{\text{comm}}}{N}$, and the average number of reassignment iterations $I$, respectively, for the two work exchange schemes versus  $\sigma^2$. For each value of $\sigma^2$, we generate 50 random values of the heterogeneity set $\bar{\lambda}$ , and find the average quantities for each. The points falling on the solid lines are the average values over the 50 readings, and the vertical lines represent the error bars. We notice the following: 
1) With  heterogeneity knowledge, the average extra communication is almost zero for any $\sigma^2$, and that is due to the initial assignment in (\ref{eq:assign-WC}) which is the expected number of points to be done by each worker, and hence negligible extra communication is needed.
2) Without  heterogeneity knowledge, the average extra communication is zero for zero variance and that is according to (\ref{eq:expected-comm-sch2}), where for $\sigma^2=0$ we have $\lambda_k/\lambda_{\text{sum}}=1/K$, and hence $E[N_{\text{comm}}]=0$, and as $\sigma^2$ increases $E[N_{\text{comm}}]$ increases.
 3) The two proposed schemes have an increasing values of $I$ as $\sigma^2$ increases, with slight lower values for the known heterogeneity case.

\begin{figure}[h!]
  \begin{center}
  \includegraphics[width=0.5\columnwidth]{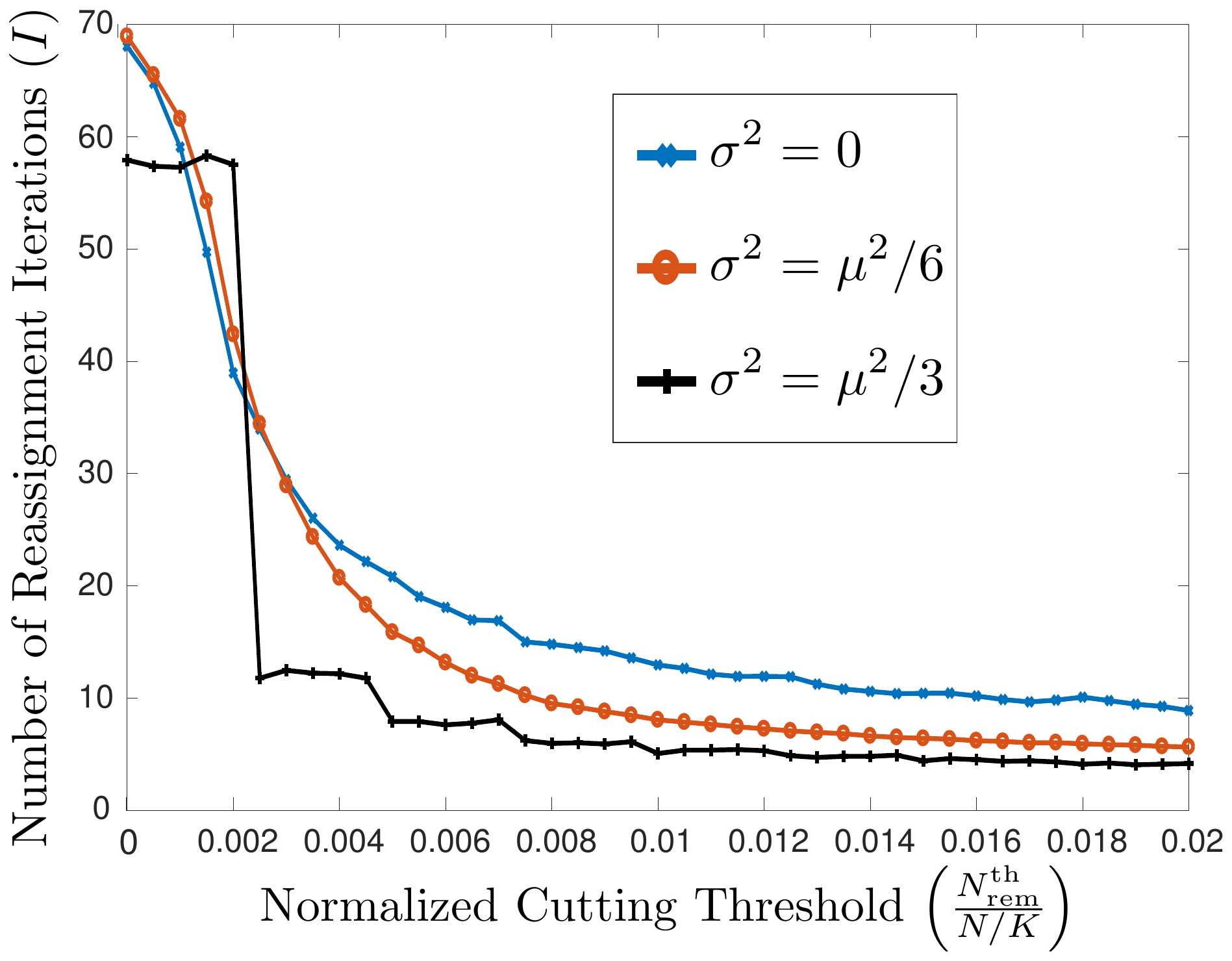}
\caption{Number of reassignments $(I)$ versus the cutting threshold $(N_{\text{rem}}^{\text{th}})$ for the work exchange scheme with unknown heterogeneity, and different values of  $\sigma^2$.\label{fig:num4}}
  \end{center}
  \vspace{-20pt}
\end{figure}

\vspace{10pt}
\noindent\textbf{Effect of adding a cutting threshold $N_{\text{rem}}^{\text{th}}$ on $I$:}
In case of unknown heterogeneity, we set $\mu=50$ and vary $N_{\text{rem}}^{\text{th}}$ in Figure~\ref{fig:num4}. We can see that as $N_{\text{rem}}^{\text{th}}$ increases, the value of $I$ decreases for different values of $\sigma^2$ as shown in the figure. Furthermore, we notice that after a certain point, the value of $I$ starts to saturate. Therefore, we set the default value of $N_{\text{rem}}^{\text{th}}$ to  $0.01 \frac{N}{K}$ for our example of $K=50$ workers, which is sufficient to get small mean number of reassignment iterations $I$ without compromising the mean computation time $T_{\text{comp}}$ as we can see in Figure~\ref{fig:num1} ($T_{\text{comp}}$ is still very close to $T_{\text{comp}}^{\text{oracle}}$).

\section{Conclusion and Future Directions}

In this paper, we focused on the problem of mitigating the impact of computational heterogeneity in large-scale distributed computation. 
We first obtained an oracle lower bound on the expected computation time as a function of the heterogeneity parameters, which has an interesting water-filling interpretation.  Inspired by the goal of approaching the oracle lower bound, we presented the idea of work exchange, in which the master node can perform iterative reassignments of computational tasks to the workers. We presented two variations of the work-exchange idea: the first in which heterogeneity knowledge is available; and a  second more realistic variation when the heterogeneity knowledge is unknown apriori, and must be estimated in an online manner.  These schemes reveal that there are fundamental tradeoffs between latency (time for computation), level of coordination and  communication between the master node and the workers.
Interestingly, it was found via simulation results that the computation time of the work-exchange scheme is very close to the oracle lower bound even when the heterogeneity is unknown apriori, in contrast to the baseline MDS coded computation scheme when the heterogeneity levels are high.
Moreover, the work exchange approach does not require linearity of the underlying computational task (unlike coded computation approaches), but rather assumes flexible divisions of the work
into smaller independent sub-tasks.

There are several interesting directions for future work which include 
a) adapting the work-exchange algorithms for sequential computations where the smaller sub-tasks depend on each other as in gradient descent and its stochastic variations;
b) imposing divisibility constraints where the computational task can only be divided into bounded number of sub-tasks; and
 c) obtain mathematical guarantees for the proposed work exchange schemes in terms of the expected latency, as a function of the heterogeneity levels, coordination and communication overheads. 

\label{sec:conclusion}

\bibliographystyle{IEEEtran}
\bibliography{./refs}

\end{document}